\documentclass[a4paper]{article}

\usepackage{amsmath}
\usepackage{amssymb}
\usepackage{graphicx}
\usepackage[OT4]{fontenc}
\usepackage[latin2]{inputenc}
\usepackage{tikz}
\usepackage{url}
\usepackage[numbers]{natbib}
\usepackage{bm}
\usepackage{setspace}

\newcommand{\abs}[1]{\left| #1 \right|}

\newcommand{\okra}[1]{\left( #1 \right)}
\newcommand{\kwad}[1]{\left[ #1 \right]}
\newcommand{\klam}[1]{\left\{ #1 \right\}}
\newcommand{\ceil}[1]{\left\lceil #1 \right\rceil}
\newcommand{\floor}[1]{\left\lfloor #1 \right\rfloor}

\DeclareMathOperator{\card}{card}

\DeclareMathOperator{\sred}{\mathbf{E}}

\newcommand{\boole}[1]{{\bf 1}{\klam{#1}}}

\newtheorem{definition}{Definition}

\newtheorem{proposition}{Proposition}
\newtheorem{theorem}{Theorem}

\newenvironment*{proof}{\begin{trivlist}\item[]
\noindent\textbf{Proof:}}{$\Box$\par\end{trivlist}}
\newenvironment*{proof*}[1]{\begin{trivlist}\item[]
\noindent\textbf{Proof of #1:}}{$\Box$\par\end{trivlist}}

\author{{\L}ukasz D\k{e}bowski\thanks{
    {\L}. D\k{e}bowski is with
    the Institute of Computer Science, Polish Academy of Sciences, 
    ul. Jana Kazimierza 5, 01-248 Warszawa, Poland 
    (e-mail: ldebowsk@ipipan.waw.pl).
    \newline\null $\quad$ 
    An extended abstract of this paper,
    entitled ``Regular Hilberg Processes: Nonexistence of Universal
    Redundancy Ratios'', was presented at the 8th Workshop on
    Information Theoretic Methods in Science and Engineering,
    Copenhagen, June 24-26, 2015.}  }

\title{Regular Hilberg Processes: An Example of Processes with a
  Vanishing Entropy Rate} \date{}

\begin{document}

\pagestyle{empty}   
\begin{titlepage}
\maketitle

\begin{abstract}
  A regular Hilberg process is a stationary process that satisfies
  both a hyperlogarithmic growth of maximal repetition and a power-law
  growth of topological entropy, which are a kind of dual
  conditions. The hyperlogarithmic growth of maximal repetition has
  been experimentally observed for texts in natural language, whereas
  the power-law growth of topological entropy implies a vanishing
  Shannon entropy rate and thus probably does not hold for natural
  language.  In this paper, we provide a constructive example of
  regular Hilberg processes, which we call random hierarchical
  association (RHA) processes. Our construction does not apply the
  standard cutting and stacking method. For the constructed RHA
  processes, we demonstrate that the expected length of any uniquely
  decodable code is orders of magnitude larger than the Shannon block
  entropy of the ergodic component of the RHA process. Our proposition
  supplements the classical result by Shields concerning nonexistence
  of universal redundancy rates.
%
  \\[1em]
\end{abstract}


\end{titlepage}
\pagestyle{plain}   


\section{Main ideas and results}
\label{secIntroduction}

Throughout this paper we identify stationary processes with their
distributions (stationary measures) and we use terms ``measure'' and
``process'' interchangeably. Consider thus a stationary measure $\mu$
on the measurable space of infinite sequences
$(\mathbb{A}^{\mathbb{N}},\mathcal{A}^{\mathbb{N}})$ from a finite
alphabet $\mathbb{A}\subset \mathbb{N}$. The random symbols will be
denoted as
$\xi_i:\mathbb{A}^{\mathbb{N}}\ni (x_i)_{i\in\mathbb{N}}\mapsto
x_i\in\mathbb{A}$,
whereas blocks of symbols will be denoted as
$x_{k:l}=(x_i)_{i=k}^l$. The expectation with respect to $\mu$ is
denoted as $\sred_\mu$. We also use shorthand
$\mu(x_{1:m})=\mu(\xi_{1:m}=x_{1:m})$. The Shannon block entropy of
measure $\mu$ is function
\begin{align}
  \label{BlockEntropy}
  H_\mu(m):=\sred_\mu \kwad{-\log \mu(\xi_{1:m})}
  ,
\end{align}
and the Shannon entropy rate of $\mu$ is the limit
\begin{align}
  \label{EntropyRate}
  h_\mu:=\inf_{m\in\mathbb{N}}\frac{H_\mu(m)}{m}=
  \lim_{m\rightarrow\infty} \frac{H_\mu(m)}{m}
  .
\end{align}

Let us introduce two functions of an individual block $\xi_{1:k}$. The
first one is the maximal repetition
\begin{align}
  L(\xi_{1:k})&:=\max \klam{m: \text{some $x_{1:m}$ is repeated in
      $\xi_{1:k}$}}
\end{align}
\cite{DeLuca99,Shields92b,Shields97,KolpakovKucherov99a,KolpakovKucherov99},
whereas the second one is the topological entropy
\begin{align}
  H_{top}(m|\xi_{1:k}):=\log \card\klam{x_{1:m}:
    \text{$x_{1:m}$ is a substring of $\xi_{1:k}$}}
  ,
\end{align}
which is the logarithm of subword complexity
\cite{JansonLonardiSzpankowski04, Ferenczi99, DeLuca99,
  GheorghiciucWard07, Ivanko08, Debowski16}.  In this paper, we are
interested in the following class of stationary processes, defined
using the Big O notation:
\begin{definition}[a variation of a definition in \cite{Debowski14f}]
  \label{defiRHP}
  A stationary measure $\mu$ on the measurable space of infinite
  sequences $(\mathbb{A}^{\mathbb{N}},\mathcal{A}^{\mathbb{N}})$ is
  called a \emph{regular Hilberg process with an exponent
    $\beta\in(0,1)$} if it satisfies conditions
\begin{align}
  \label{HilbergRepetitionAS}
  L(\xi_{1:m}) 
  &=\Theta\okra{ (\log m)^{1/\beta} }
  ,
  \\
  \label{HilbergTopEntropyAS}
  H_{top}(m|\xi_{1:\infty})
  &=\Theta\okra{ m^{\beta} }
  .
\end{align}
  $\mu$-almost surely, where the lower bound for the maximal
  repetition and the upper bound for the topological entropy are
  uniform in $\xi_{1:\infty}$.
\end{definition}

The original definition in \cite{Debowski14f} uses condition
$H_\mu(m)=\Theta\okra{ m^\beta }$ rather than
(\ref{HilbergTopEntropyAS}) and condition $\sred_\mu L(\xi_{1:m})
=\Theta\okra{ (\log m)^{1/\beta} }$ instead of
(\ref{HilbergRepetitionAS}). Condition $H_\mu(m)=\Theta\okra{ m^\beta
}$ has been originally contemplated by Hilberg \cite{Hilberg90}, hence
follows the name of the class of processes. Conditions
(\ref{HilbergRepetitionAS}) and (\ref{HilbergTopEntropyAS}) are,
however, more natural since they pertain to an individual sequence
$\xi_{1:\infty}$ and are dual in view of the following proposition:
\begin{theorem}[\cite{Debowski15f}]
  \label{theoTopEntropyRepetition}
  If $H_{top}(m|\xi_{1:k})< \log (k-m+1)$ then
  $L(\xi_{1:k})\ge m$.
\end{theorem}
\begin{proof}
  String $\xi_{1:k}$ contains $k-m+1$ substrings of length $m$ (on
  overlapping positions). Among them there can be at most $\exp
  (H_{top}(m|\xi_{1:k}))$ different substrings.  Since $\exp
  (H_{top}(m|\xi_{1:k}))< k-m+1$, there must be some repeat of length
  $m$. Hence $L(\xi_{1:k})\ge m$.
\end{proof}  
In particular, since
$H_{top}(m|\xi_{1:k})\le H_{top}(m|\xi_{1:\infty})$, Theorem
\ref{theoTopEntropyRepetition} yields
\begin{align*}
  H_{top}(m|\xi_{1:\infty})
  =O\okra{ m^{\beta} }
  \Rightarrow
  L(\xi_{1:m}) 
  =\Omega\okra{ (\log m)^{1/\beta} }
  ,
  \\
  L(\xi_{1:m}) 
  =O\okra{ (\log m)^{1/\beta} }
  \Rightarrow
  H_{top}(m|\xi_{1:\infty})
  =\Omega\okra{ m^{\beta} }
  .
\end{align*}
Now we can see that the lower bound in (\ref{HilbergRepetitionAS}) is
implied by the upper bound in (\ref{HilbergTopEntropyAS}), whereas the
upper bound in (\ref{HilbergRepetitionAS}) implies the lower bound in
(\ref{HilbergTopEntropyAS}). We might therefore suppose that
conditions (\ref{HilbergRepetitionAS}) and (\ref{HilbergTopEntropyAS})
hold simultaneously indeed for some class of processes.

Why is this problem important?  In fact, according to some
experimental measurements of maximal repetition, the hyperlogarithmic
growth (\ref{HilbergRepetitionAS}) holds approximately with
$\beta\approx 0.4$ for texts in English, French, and German, where the
lower bound for the growth of maximal repetition seems uniform, i.e.,
text-independent \cite{Debowski12b,Debowski15f}.  Thus understanding
how to construct some class of processes satisfying condition
(\ref{HilbergRepetitionAS}) may contribute to an improvement in
statistical models of natural language. Although condition
$H_\mu(m)=\Theta\okra{ m^\beta }$, related to
(\ref{HilbergTopEntropyAS}), was actually considered in
\cite{Hilberg90} as a hypothesis for natural language, here we should
admit that the combination of conditions (\ref{HilbergRepetitionAS})
and (\ref{HilbergTopEntropyAS}) is likely too strong to be required
from the natural language models. As we will show, the power law
(\ref{HilbergTopEntropyAS}) implies a vanishing Shannon entropy rate,
$h_\mu=0$, whereas the overwhelming empirical evidence asserts that
the Shannon entropy rate of natural language is strictly positive,
about $1$ bit per character
\cite{Shannon51,CoverKing78,BrownOthers83,Grassberger02,BehrOthers03,TakahiraOthers16}.
Nevertheless, constructing stationary processes that satisfy the
hyperlogarithmic growth (\ref{HilbergRepetitionAS}) is nontrivial
enough, so it may be illuminating to consider first a somewhat
unrealistic class of processes that also satisfy the power law
(\ref{HilbergTopEntropyAS}).

For the regular Hilberg processes there are two general results.  As
mentioned, it can be seen easily that the power law
(\ref{HilbergTopEntropyAS}) implies a vanishing Shannon entropy rate.
\begin{theorem}
  \label{theoHilbergEntropyRate}
  We have $h_\mu=0$ for a regular Hilberg process $\mu$.
\end{theorem}
\begin{proof}
  The argument involves the random ergodic measure
  $F=\mu(\cdot|\mathcal{I})$, where $\mathcal{I}$ is the
  shift-invariant algebra \cite{Kallenberg97,GrayDavisson74b}. By the
  ergodic theorem for stationary processes \cite{Kallenberg97}, we
  have $\mu$-almost surely
\begin{align}
  \label{HFHTop}
    H_{top}(m|\xi_{1:\infty})\ge \log
    \card\klam{x_{1:m}:F(x_{1:m})>0}\ge H_F(m) 
    ,
\end{align}
so $h_F=0$ follows from (\ref{HilbergTopEntropyAS}), whereas as shown
in \cite{GrayDavisson74b,Debowski09} we have
\begin{align}
  h_\mu=\sred_\mu h_F
  ,  
\end{align}
from which $h_\mu=0$ follows.
\end{proof}

Moreover, the ergodic decomposition of a regular Hilberg
process, as defined in Definition~\ref{defiRHP}, consists of ergodic
regular Hilberg processes. Namely, we have:
\begin{theorem}
  For a regular Hilberg process $\mu$ with exponent $\beta$, the
  random ergodic measure $F=\mu(\cdot|\mathcal{I})$, where
  $\mathcal{I}$ is the shift-invariant algebra, $\mu$-almost surely constitutes
  an ergodic regular Hilberg process with exponent $\beta$.
\end{theorem}
\begin{proof}
  We have $\mu=\int F d\mu$. Hence every event of full measure $\mu$
  must be $\mu$-almost surely an event of full measure $F$. This
  implies the claim.
\end{proof}
We suppose that the above property is not true for the original
definition of a regular Hilberg process given in article
\cite{Debowski14f}, but we do not investigate this problem in this
paper.

We will present now some constructive example of regular Hilberg
processes. The example will be called random hierarchical association
(RHA) processes.  The RHA processes are parameterized by certain free
parameters which we will call perplexities (a name borrowed from
computational linguistics). Approximately, perplexity $k_n$ is the
number of distinct blocks of length $2^n$ that appear in the process
realization. Exactly in this meaning, term ``perplexity'' is used in
computational linguistics. It turns out that controlling perplexities,
we can control the value of the Shannon block entropy and force the
Shannon entropy rate to be zero. It turns out as well that we can
control the value of the topological entropy and the maximal
repetition. In this way we can construct a stationary process
exhibiting quite an arbitrary desired growth of the topological
entropy and the maximal repetition, such a regular Hilberg process.

We have invented the RHA processes as a construction unrelated to the
cutting and stacking method \cite{Shields91b}, used for constructing
stationary processes with certain desired properties. The cutting and
stacking method seems more abstract and more general than the RHA
process method. Certainly, these two methods adopt very different
strategies. The cutting and stacking method, being a tool borrowed
from ergodic theory, approximates the constructed process by an
abstract dynamical system. This dynamical system consists of the
Lebesgue measure on the unit interval with an incrementally
constructed partition and transformation. In contrast, the RHA process
method begins with some nonstationary nonergodic process from which we
obtain a given stationary ergodic measure by taking the stationary
mean and ergodic decomposition. For our particular application of
constructing regular Hilberg processes, the RHA process method is
sufficient and seems natural enough but it is likely insufficient for
constructing processes which satisfy condition
(\ref{HilbergRepetitionAS}) without condition
(\ref{HilbergTopEntropyAS}). In the later case, being the case of
interest for modeling natural language, using the cutting and stacking
method is a certain idea but we have not figured out yet how to
implement it exactly.

To briefly explain our method, the RHA processes are formed in two not
so complicated steps.  First, we sample recursively random pools of
$k_n$ distinct blocks of length $2^n$, which are formed by
concatenation of randomly selected $k_n$ pairs chosen from $k_{n-1}$
distinct blocks of length $2^{n-1}$ sampled in a previous step (the
recursion stops at blocks of length $1$, which are fixed
symbols). Second, we obtain an infinite sequence of random symbols by
concatenating blocks of lengths $2^0$, $2^1$, $2^2$, $...$ randomly
chosen from the respective pools.  As a result there cannot be more
that $k_n^2$ distinct blocks of length $2^n$ that appear the final
process realization. The selection of these blocks is, however, random
and we do not know them a priori. This is some reason why the
constructed process satisfies conditions similar to
(\ref{HilbergRepetitionAS}) and (\ref{HilbergTopEntropyAS})
simultaneously but is nonergodic.

Now we will write down this construction using symbols.

\textbf{Step 1:} Formally, let perplexities
$(k_n)_{n\in\klam{0}\cup\mathbb{N}}$ be some sequence of strictly
positive natural numbers that satisfy
\begin{align}
  \label{LogSubadditivity}
  k_{n-1}\le k_n \le k_{n-1}^2
  .
\end{align}
Next, for each $n\in\mathbb{N}$, let
$(L_{nj},R_{nj})_{j\in\klam{1,...,k_{n}}}$ be an independent random
combination of $k_n$ pairs of numbers from the set
$\klam{1,...,k_{n-1}}$ drawn without repetition. That is, we assume
that each pair $(L_{nj},R_{nj})$ is different, the elements of pairs
may be identical ($L_{nj}=R_{nj}$), and the sequence
$(L_{nj},R_{nj})_{j\in\klam{1,...,k_{n}}}$ is sorted
lexicographically. Formally, we assume that random variables $L_{nj}$
and $R_{nj}$ are supported on some probability space
$(\Omega,\mathcal{J},P)$ and have the uniform distribution
\begin{align}
  &P((L_{n1},R_{n1},...,L_{nk_{n}},R_{nk_{n}})=(l_{n1},r_{n1},...,l_{nk_{n}},r_{nk_{n}}))
  \nonumber\\
  &\phantom{==================}=
  \binom{k_{n-1}^2}{k_n}^{-1}
  .
\end{align}
Subsequently we define random variables
\begin{align}
  Y^0_j&=j, & j&\in \klam{1,...,k_{0}},\\
  Y^n_j&=Y^{n-1}_{L_{nj}}\times Y^{n-1}_{R_{nj}}, & j&\in
  \klam{1,...,k_{n}}, n\in\mathbb{N}
  ,
\end{align}
where $a\times b$ denotes concatenation. Hence $Y^n_j$ are $k_n$
distinct blocks of $2^n$ natural numbers, selected by some sort of
random hierarchical concatenation.

\textbf{Step 2:} Variables $Y^n_j$ will be the building blocks of yet
another process.  Let $(C_{n})_{n\in\klam{0}\cup\mathbb{N}}$ be
independent random variables, independent from
$(L_{nj},R_{nj})_{n\in\mathbb{N},j\in\klam{1,...,k_{n}}}$, with
uniform distribution
\begin{align}
  P(C_{n}=j)&=1/k_{n},  &j&\in\klam{1,...,k_{n}}
  .
\end{align}
\begin{definition}
  \label{defiRHA}
  The \emph{random hierarchical association (RHA) process
    $\mathcal{X}$ with perplexities $(k_n)_{n\in
      \klam{0}\cup\mathbb{N}}$} is defined as
\begin{align}
  \mathcal{X}=Y^0_{C_0}\times Y^1_{C_1}\times Y^2_{C_2}\times ...\, .
\end{align}
\end{definition}
This completes the construction of the RHA processes but it is not the
end of our discussion of these processes.

It is convenient to define a few more random variables for the RHA
process.  First, sequence $\mathcal{X}$ will be parsed into a sequence
of numbers $X_j$, where
\begin{align}
  \label{ParsingOfRHA}
  \mathcal{X}=X_1\times X_2\times X_3\times ... 
  ,
\end{align}
and, second, we denote blocks starting at any position as
\begin{align}
  \label{BlocksOfRHA}
  X_{k:l}=X_k\times X_{k+1}\times... \times X_l 
  .
\end{align}

The RHA processes defined in Definition \ref{defiRHA} are not
stationary but they possess a stationary mean, which is a condition
related to asymptotic mean stationarity.  Let us introduce shift
operation $T:\mathbb{A}^{\mathbb{N}}\ni (x_i)_{i\in\mathbb{N}}\mapsto
(x_{i+1})_{i\in\mathbb{N}}\in\mathbb{A}^{\mathbb{N}}$.  We recall this
definition:
\begin{definition}
  A measure $\nu$ on
  $(\mathbb{A}^{\mathbb{N}},\mathcal{A}^{\mathbb{N}})$ is called
  \emph{asymptotically mean stationary (AMS)} if limits
  \begin{align}
    \mu(A):=
    \lim_{N\rightarrow\infty}\frac{1}{N}
    \sum_{i=1}^{N}\nu(T^{-i}A)
  \end{align}
  exist for every event $A\in\mathcal{A}^{\mathbb{N}}$ \cite{GrayKieffer80}.
\end{definition}
For an AMS measure $\nu$, function $\mu$ is a stationary measure on
$(\mathbb{A}^{\mathbb{N}},\mathcal{A}^{\mathbb{N}})$, called the
stationary mean of $\nu$. Moreover, measures $\mu$ and $\nu$ are equal
on the shift invariant algebra
$\mathcal{I}=\klam{A\in\mathcal{A}^{\mathbb{N}}: T^{-1}A =A}$, i.e.,
$\mu(A)=\nu(A)$ for all $A\in\mathcal{I}$.

Now, let $\mathbb{A}^{+}=\bigcup_{n\in\mathbb{N}} \mathbb{A}^{n}$.
There is a related relaxed condition of asymptotic mean stationarity:
\begin{definition}
  A measure $\nu$ on
  $(\mathbb{A}^{\mathbb{N}},\mathcal{A}^{\mathbb{N}})$ is called
  \emph{pseudo-asymptotically mean stationary (pseudo-AMS)}
  if limits
  \begin{align}
    \mu(x_{1:m}):=
    \lim_{N\rightarrow\infty}\frac{1}{N}
    \sum_{i=1}^{N}\nu(\xi_{i:i+m-1}=x_{1:m})
  \end{align}
  exist for every block $x_{1:m}\in\mathbb{A}^{+}$.
\end{definition}
For a pseudo-AMS measure $\nu$ over a finite alphabet $\mathbb{A}$,
function $\mu$, extended via $\mu(\xi_{1:m}=x_{1:m}):=\mu(x_{1:m})$,
is also a stationary measure on
$(\mathbb{A}^{\mathbb{N}},\mathcal{A}^{\mathbb{N}})$.  We shall
continue to call this $\mu$ a stationary mean of $\nu$. However, a
pseudo-AMS measure need not be AMS in general, cf.\ \cite[Remark in
the proof of Lemma 7.16]{Gray09} and \cite[Example
6.3]{Debowski10}. In particular, for a pseudo-AMS measure $\nu$ we
need not have $\mu(A)=\nu(A)$ for shift invariant events
$A\in\mathcal{I}$.

It turns out that the RHA processes are pseudo-AMS.
\begin{theorem}
  \label{theoAMS}
  The RHA processes are pseudo-AMS. In particular, for $m\le 2^n$ and
  $k\in\mathbb{N}$, the stationary mean is
  \begin{align}
    \mu(x_{1:m})= \frac{1}{2^n}\sum_{j=0}^{2^n-1}
    P(X_{k2^n+j:k2^n+j+m-1}=x_{1:m}) .
  \end{align}
\end{theorem}
The proof of Theorem \ref{theoAMS} will be presented later in this
article.

We suppose that the RHA processes are also AMS but we could not prove
it so far. However, we have been able to show that certain RHA
processes give rise to regular Hilberg processes:
\begin{theorem}
  \label{theoMain}
  For perplexities 
  \begin{align}
    \label{HilbergPerplexity}
    k_n=\floor{\exp\okra{2^{\beta n}}}
    ,
  \end{align}
  where $0< \beta< 1$, the stationary mean $\mu$ of the RHA process 
  satisfies the following conditions:
  \begin{enumerate}
  \item The Shannon entropy rate is $h_\mu=0$.
  \item The Shannon block entropy is sandwiched by
    \begin{align}
      \label{HilbergPerplexityEntropy}
      \frac{C_1 m}{(\log m)^{\alpha}}
      \le H_{\mu}(m)\le
      C_2 m\okra{\frac{\log\log m}{\log m}}^{\alpha}
      ,
    \end{align}
    where $\alpha=1/\beta-1$.
  \item The stationary mean $\mu$ is a regular Hilberg process with
    exponent $\beta$.
  \item The stationary mean $\mu$ is nonergodic and the Shannon
    entropy of the shift invariant algebra $H_\mu(\mathcal{I})$, as
    defined in \cite{Debowski09}, is infinite.
  \end{enumerate}
\end{theorem}
The proof of Theorem \ref{theoMain}, which we consider the main result
of this paper, will be postponed, as well. Although claim (i) follows
from claim (iii) by Theorem \ref{theoHilbergEntropyRate}, it will be
established using a different method, of an independent interest.

Theorem \ref{theoMain} has some implications for universal coding.
For a uniquely decodable code $C$, we denote its length for block
$\xi_{1:m}$ as $\abs{C(\xi_{1:m})}$. We recall that
$\sred_\mu\abs{C(\xi_{1:m})}\ge H_\mu(m)$, so the Shannon block
entropy provides a lower bound for compression of a stochastic
process.  In contrast, a code $C$ is called universal if
\begin{align}
  \label{Universal}
  \lim_{m\rightarrow\infty}\frac{\abs{C(\xi_{1:m})}}{m}=h_\mu
\end{align}
holds almost surely for every stationary ergodic measure
$\mu$. Universal codes exist and the Lempel-Ziv code
\cite{ZivLempel77} is some example of such a code. The convergence
rate for universal codes can be arbitrarily slow, however.  Shields
\cite{Shields93} showed that for any uniquely decodable code $C$ and
any sublinear function $\rho(m)=o(m)$ there exists such an ergodic
source $\mu$ that
\begin{align}
  \limsup_{m\rightarrow\infty}[\sred_\mu \abs{C(\xi_{1:m})}-H_\mu(m)-\rho(m)]>0.
\end{align}

Whereas Shields' result concerns nonexistence of a universal sublinear
bound for the difference $\abs{C(\xi_{1:m})}-H_\mu(m)$, some way of
supplementing it is to investigate ratio
$\abs{C(\xi_{1:m})}/H_\mu(m)$. Although this ratio is asymptotically
equal to $1$ for universal codes and processes with a positive Shannon
entropy rate $h_\mu>0$, Shields' result does not predict how the ratio
behaves for processes with a vanishing Shannon entropy rate $h_\mu=0$.
In fact, for the Lempel-Ziv code and ergodic regular Hilberg
processes, there is no essentially sublinear bound for the ratio
$\abs{C(\xi_{1:m})}/H_\mu(m)$:
\begin{theorem}
  \label{theoLZRatio}
  Let $C$ be the Lempel-Ziv code. For an ergodic regular Hilberg
  process $\mu$ with exponent $\beta$, $\mu$-almost surely
  \begin{align}
    \label{LZRatio}
    \frac{\abs{C(\xi_{1:m})}}{H_\mu(m)}=\Omega\okra{
      \frac{m^{1-\beta}}{(\log m)^{1/\beta-1}} }
    .
  \end{align}
\end{theorem}
\begin{proof} By ergodicity, we have $\mu=F$. Thus, by (\ref{HFHTop})
  and (\ref{HilbergTopEntropyAS}), we obtain
\begin{align}
  H_\mu(m)=H_F(m)\le H_{top}(m|\xi_{1:\infty}) =O\okra{m^\beta} 
  .
\end{align}
On the other hand, the length of the Lempel-Ziv code
$\abs{C(\xi_{1:m})}$ for a block $\xi_{1:m}$, by
(\ref{HilbergRepetitionAS}), $\mu$-almost surely satisfies
\begin{align}
  \label{LZCompressionBound}
  \abs{C(\xi_{1:m})}&\ge \frac{m}{L(\xi_{1:m})+1}\log
  \frac{m}{L(\xi_{1:m})+1}
  \nonumber\\
  &= \Omega\okra{
  \frac{m}{(\log m)^{1/\beta-1}} } 
  .
\end{align}
The first inequality in (\ref{LZCompressionBound}) stems from a simple
observation in \cite{Debowski14f} that the length of the Lempel-Ziv
code is greater than $V\log V$, where $V$ is the number of Lempel-Ziv
phrases, whereas the Lempel-Ziv phrases may not be longer than the
maximal repetition plus $1$.
\end{proof}

A somewhat more general result holds for the RHA processes from
Theorem~\ref{theoMain}. In this case, we may replace the Lempel-Ziv
code with an arbitrary uniquely decodable code:
\begin{theorem}
  \label{theoRatio}
  Let $C$ be an arbitrary uniquely decodable code. For the stationary
  mean $\mu$ of the RHA process with perplexities
  (\ref{HilbergPerplexity}) and its random ergodic measure
  $F=\mu(\cdot|\mathcal{I})$, we have
  \begin{align}
    \label{RRatio}
    \sred_\mu\frac{\sred_F \abs{C(\xi_{1:m})}}{H_F(m)}
    &=\Omega\okra{\frac{m^{1-\beta}}{(\log m)^{1/\beta-1}}}
    ,
  \end{align}
\end{theorem}
Ratio (\ref{RRatio}) can be larger than any function
$o(m^{1-\epsilon})$.
\begin{proof} The claim follows by (\ref{HFHTop}),
  (\ref{HilbergTopEntropyAS}), (\ref{HilbergPerplexityEntropy}), and
  the source coding inequality 
  \begin{align}
    \sred_\mu\sred_F
    \abs{C(\xi_{1:m})}=\sred_\mu\abs{C(\xi_{1:m})}\ge H_\mu(m).    
  \end{align}
\end{proof}
Theorems \ref{theoLZRatio} and \ref{theoRatio} should be read as a
warning that the length of a universal code $\abs{C(\xi_{1:m})}$ is
not a very reliable estimate of the Shannon block entropy $H_\mu(m)$
for an ergodic regular Hilberg process. Whereas, using a universal
code, we can reliably estimate the Shannon entropy rate $h_\mu$, the
code length $\abs{C(\xi_{1:m})}$ can be orders of magnitude larger
than the Shannon block entropy $H_\mu(m)$.

The remaining parts of this article are devoted to proving the more
involved Theorems \ref{theoAMS} and \ref{theoMain}.  The organization
is as follows.  In Section \ref{secDefinition}, some auxiliary
notations are introduced.  In Section \ref{secMean}, Theorem
\ref{theoAMS} is demonstrated.  In Section \ref{secFromRHA}, the
entropies and the maximal repetition for the RHA process and its
stationary mean are related.  Section \ref{secNoRepeat} concerns some
further auxiliary results, such as probabilities of no repeat and a
bound for the topological entropy.  In Section \ref{secEntropy},
Shannon block entropies of the RHA processes are discussed.  In
Section \ref{secMain}, Theorem \ref{theoMain} is proved.

\section{Auxiliary notations}
\label{secDefinition}

Let us recall the construction of the RHA process from the previous
section.  In this section we introduce a few notations which will be
used further. The collection of random variables $(L_{nj},R_{nj})$
will be denoted as
\begin{align}
  \mathcal{G}=(L_{nj},R_{nj})_{n\in\mathbb{N},j\in\klam{1,...,k_{n}}}
  .
\end{align}
We will also use notations
\begin{align}
  \mathcal{G}_{\le m}&=(L_{nj},R_{nj})_{n\le m,j\in\klam{1,...,k_{n}}}
  ,
  \\
  \mathcal{G}_{>m}&=(L_{nj},R_{nj})_{n> m,j\in\klam{1,...,k_{n}}}
  .
\end{align}
Let us observe that collection $\mathcal{G}_{\le m}$ fully determines
variables $Y^m_j$ for a fixed $m$.

It is convenient to define a few more random variables for the RHA
process.  First, generalizing parsing (\ref{ParsingOfRHA}), sequence
$\mathcal{X}$ will be parsed into a sequence of blocks $X^n_j$ of
length $2^n$, where
\begin{align}
  \mathcal{X}=Y^1_{C_0}\times Y^1_{C_1}\times Y^2_{C_2}\times ...
  \times Y^n_{C_n}(=X^n_1)\times X^n_2\times X^n_3\times ... \, .
\end{align}
Let us also observe that there exist unique random variables $K_{nj}$
such that
\begin{align}
  X^n_j=Y^n_{K_{nj}}
  .
\end{align}
Moreover, generalizing notation (\ref{BlocksOfRHA}), we also denote
blocks of length $2^n$ starting at any position as
\begin{align}
  X^n_{k:l}=X^n_k\times X^n_{k+1}\times... \times X^n_l
  .
\end{align}

\section{Stationary mean}
\label{secMean}

In this section, we will demonstrate Theorem \ref{theoAMS}. This
theorem states that the RHA process has a stationary mean in a weaker
sense, i.e., it is pseudo-asymptotically mean stationary (pseudo-AMS).

First we will prove this useful and a bit surprising property, which
will be used in the present and in the further sections.
\begin{proposition}
  \label{theoKpairs}
  Variables $K_{nj}$ are independent from $\mathcal{G}_{\le n}$ and
  satisfy
  \begin{align}
    \label{Kpairs}
    P(K_{nj}=l,K_{n,j+1}=m)&=1/k_n^2, & l,m&\in\klam{1,...,k_{n}},
    j\in\mathbb{N}
    .
  \end{align}
\end{proposition}
\begin{proof}
  Each $K_{nj}$ is a function of $C_{q}$ for some $q\ge n$ and
  $\mathcal{G}_{> n}$. Hence $K_{nj}$ are independent from
  $\mathcal{G}_{\le n}$.

  Now we will show by induction on $j$ that (\ref{Kpairs}) is
  satisfied. 
  
  The induction begins with $K_{n1}=C_n$ and
  $K_{n2}=L_{n+1,C_{n+1}}$. These two variables are independent by
  definition and by definition $K_{n1}$ is uniformly distributed on
  $\klam{1,...,k_{n}}$. It remains to show that so is
  $K_{n2}$. Observe that $(L_{n+1,k},R_{n+1,k})$ are independent of
  $C_{n+1}$. Hence for $l,m\in\klam{1,...,k_{n}}$ we obtain
  \begin{align*}
    P(K_{n2}=l,K_{n3}=m)&=
    \sum_{k=1}^{k_{n+1}}P(L_{n+1,k}=l,R_{n+1,k}=m)P(C_{n+1}=k)
    \\
    &=
    \frac{1}{k_{n+1}}\sum_{k=1}^{k_{n+1}}P(L_{n+1,k}=l,R_{n+1,k}=m)
    \\
    &=
    \frac{1}{k_{n+1}}\binom{k_{n}^2}{k_{n+1}}^{-1}\binom{k_{n}^2-1}{k_{n+1}-1}
    =
    \frac{1}{k_{n+1}}\frac{k_{n+1}}{k_{n}^2}=\frac{1}{k_{n}^2}
    ,  
  \end{align*}
  so $K_{n2}$ is uniformly distributed on $\klam{1,...,k_{n}}$.
  
  The inductive step is as follows: (i) if $K_{n+1,j}$ is uniformly
  distributed on $\klam{1,...,k_{n+1}}$ then
  $(K_{n,2j},K_{n,2j+1})=(L_{n+1,K_{n+1,j}}, R_{n+1,K_{n+1,j}})$ is
  uniformly distributed on $\klam{1,...,k_{n}}\times
  \klam{1,...,k_{n}}$, and (ii) if $(K_{n+1,j},K_{n+1,j+1})$ is uniformly
  distributed on $\klam{1,...,k_{n+1}}\times \klam{1,...,k_{n+1}}$
  then $(K_{n,2j+1},K_{n,2j+2})=(R_{n+1,K_{n+1,j}},
  L_{n+1,K_{n+1,j+1}})$ is uniformly distributed on
  $\klam{1,...,k_{n}}\times \klam{1,...,k_{n}}$. Now
  observe that $(L_{n+1,k},R_{n+1,k})$ are independent of
  $K_{n+1,j}$. Hence, for $l,m\in\klam{1,...,k_{n}}$ we obtain
  \begin{align*}
    &P(K_{n,2j}=l,K_{n,2j+1}=m)
    \\
    &=
    \sum_{k=1}^{k_{n+1}}P(L_{n+1,k}=l,R_{n+1,k}=m)P(K_{n+1,j}=k)
    \\
    &=
    \frac{1}{k_{n+1}}\sum_{k=1}^{k_{n+1}}P(L_{n+1,k}=l,R_{n+1,k}=m)
    \\
    &=
    \frac{1}{k_{n+1}}\binom{k_{n}^2}{k_{n+1}}^{-1}\binom{k_{n}^2-1}{k_{n+1}-1}
    =
    \frac{1}{k_{n+1}}\frac{k_{n+1}}{k_{n}^2}=\frac{1}{k_{n}^2}
    ,  
  \end{align*}
  which proves claim (i). On the other hand, for
  $l,m\in\klam{1,...,k_{n}}$ we obtain
  \begin{align*}
    &P(K_{n,2j+1}=l,K_{n,2j+2}=m)
    \\
    &=
    \sum_{p,q=1}^{k_{n+1}}P(R_{n+1,p}=l,L_{n+1,q}=m)P(K_{n+1,j}=p,K_{n+1,j+1}=q)
    \\
    &=
    \frac{1}{k_{n+1}^2}\sum_{p,q=1}^{k_{n+1}}P(R_{n+1,p}=l,L_{n+1,q}=m)
    \\
    &=
    \frac{1}{k_{n+1}^2}\sum_{p=1}^{k_{n+1}}P(R_{n+1,p}=l,L_{n+1,p}=m)
    \\
    &\phantom{=}+ \frac{1}{k_{n+1}^2}\sum_{p,q=1,\,p\neq
      q}^{k_{n+1}}P(R_{n+1,p}=l,L_{n+1,q}=m)
    \\
    &= \frac{1}{k_{n+1}^2}\binom{k_{n}^2}{k_{n+1}}^{-1}
    \okra{\binom{k_{n}^2-1}{k_{n+1}-1}+
      (k_{n}^2-1)\binom{k_{n}^2-2}{k_{n+1}-2} }
    \\
    &= \frac{1}{k_{n+1}^2}
    \okra{\frac{k_{n+1}}{k_{n}^2}+
      (k_{n}^2-1)\frac{k_{n+1}(k_{n+1}-1)}{k_{n}^2(k_{n}^2-1)}}
    = \frac{1}{k_{n}^2} 
    ,
  \end{align*}
  which proves claim (ii).
\end{proof}

Using Proposition \ref{theoKpairs}, it is easy to demonstrate Theorem
\ref{theoAMS}.
\begin{proof*}{Theorem \ref{theoAMS}}
  Block $X_{k2^n+j:k2^n+j+m-1}$ is a subsequence of $X^n_{k:k+1}$ for
  $m\le 2^n$, $k\in\mathbb{N}$, and $0\le j< 2^n$. In particular,
  there exist functions $f_{mj}$ such that
  $$X_{k2^n+j:k2^n+j+m-1}=f_{mj}(X^n_{k:k+1}).$$ Hence probabilities
  $P(X_{i:i+m-1}=x_{1:m})$ are periodic functions of $i$ with period
  $2^n$, by Proposition \ref{theoKpairs}.  This implies the formula for
  $\mu(x_{1:m})$.
\end{proof*}

\section{Bounds for the stationary mean}
\label{secFromRHA}

This sections opens the discussion of various auxiliary results
necessary to establish Theorem \ref{theoMain}, the main result of this
paper. The theorem operates with three functions of the stationary
mean of the RHA process: Shannon block entropy, maximal repetition,
and topological entropy.  We first observe that it may be easier to
analyze the behavior of blocks $X^n_j$ drawn from the original the RHA
process than the behavior of its stationary mean. For this reason, in
this section we want to derive some bounds for the entropies and the
maximal repetition of the stationary mean from the analogical bounds
for blocks $X^n_j$.  In the following we will denote
\begin{align}
X^n_{kj}=X_{k2^n+j:k2^n+j+2^n-1}
.  
\end{align}
In particular, we have $X^n_{k0}=X^n_k$.

Subsequently, for Shannon entropy $H(X)=\sred_P\kwad{-\log P(X)}$, we obtain:
\begin{proposition}
  \label{theoBlockEntropyAMS}
  For the stationary mean $\mu$ of the RHA process, we have
  \begin{align}
    \label{BlockEntropyAMS}
    H(X^{n-1}_j)\le H_\mu(2^n)\le H(X^{n+1}_j)+n\log 2
    .
  \end{align}
\end{proposition}
\begin{proof}
  By the Jensen inequality for function $p\mapsto -p\log p$ and
  Theorem \ref{theoAMS}, we hence obtain
  \begin{align}
    \label{HMuGe}
    H_\mu(2^n)\ge \frac{1}{2^n}\sum_{j=0}^{2^n-1}
    H(X^n_{kj})
    .
  \end{align}
  Now we observe that for each $k\ge 1$ and $j$ there exists a $q$ such
  that $X^{n-1}_q$ is a subsequence of $X^n_{kj}$. Thus
  we have $H(X^n_{kj}) \ge H(X^{n-1}_q)$. This combined
  with inequality (\ref{HMuGe}) yields $H(X^{n-1}_j)\le H_\mu(2^n)$.
  On the other hand, using inequality $\mu(x_{1:2^n})\ge 2^{-n}
  P(X^n_{kj}=x_{1:2^n})$ and Theorem \ref{theoAMS}, we obtain
  \begin{align}
    \label{HMuLe}
    H_\mu(2^n)\le \frac{1}{2^n}\sum_{j=0}^{2^n-1}
    H(X^n_{kj}) + n\log 2
    .
  \end{align}
  Now we observe that for each $k> 1$ and $j$ there exists a $q$ such
  that $X^n_{kj}$ is a subsequence of $X^{n+1}_q$.
  Thus we have $H(X^n_{kj}) \le H(X^{n+1}_q)$. This
  combined with inequality (\ref{HMuLe}) yields $H_\mu(2^n)\le
  H(X^{n+1}_j)+n\log 2$.
\end{proof}

Analogically, we can bound the maximal repetition of the stationary
mean. The result will be stated more generally.  We will say that a
function $\phi:\mathbb{A}^+\rightarrow\mathbb{R}$ is increasing if for
$u$ being a subsequence of $w$, we have $\phi(u)\le\phi(w)$. Examples
of increasing functions include the maximal repetition $L(w)$, the
topological entropy $H_{top}(m|w)$, and the indicator function
$\boole{\phi(w)>k}$, where $\phi$ is increasing.
\begin{proposition}
  \label{theoPhiAMS}
  For the stationary mean $\mu$ of the RHA process and an increasing
  function $\phi$, we have
  \begin{align}
    \label{PhiAMS}
    \sred_P \phi(X^{n-1}_j)\le \sred_\mu \phi(\xi_{1:2^n})\le \sred_P
    \phi(X^{n+1}_j) .
  \end{align}
\end{proposition}
\begin{proof}
  By Theorem \ref{theoAMS},
  \begin{align}
    \label{PhiMu}
    \sred_\mu \phi(\xi_{1:2^n})= \frac{1}{2^n}\sum_{j=0}^{2^n-1}
    \sred_P \phi(X^n_{kj})
    .
  \end{align}
  Now we observe that for each $k\ge 1$ and $j$ there exists a $q$
  such that $X^{n-1}_q$ is a subsequence of $X^n_{kj}$. Thus we have
  $\phi(X^n_{kj}) \ge \phi(X^{n-1}_q)$. This combined with
  equality (\ref{PhiMu}) yields $\sred_P \phi(X^{n-1}_j)\le \sred_\mu
  \phi(\xi_{1:2^n})$.  On the other hand, for each $k> 1$ and $j$ there
  exists a $q$ such that $X^n_{kj}$ is a subsequence of $X^{n+1}_q$.
  Thus we have $\phi(X^n_{kj}) \le \phi(X^{n+1}_q)$. This
  combined with equality (\ref{PhiMu}) yields $\sred_\mu
  \phi(\xi_{1:2^n})\le \sred_P \phi(X^{n+1}_j)$.
\end{proof}
Hence, to obtain the desired bounds for the stationary mean, it
suffices to investigate the distribution of blocks $X^n_j$.

\section{Further auxiliary results}
\label{secNoRepeat}

To make another observation, Theorem \ref{theoMain} links the Shannon
block entropy, maximal repetition and topological entropy of the RHA
process with its parameters called perplexities $k_n$.  Therefore, the
goal of this section is to furnish some bounds for topological entropy
and maximal repetition of blocks $X^n_{kj}$ in terms of perplexities
$k_n$. In contrast, in the next section we will use perplexities $k_n$
to bound the Shannon entropies of blocks $X^n_{kj}$.

Let us begin with a simple lower bound for the topological
entropy of blocks $X^n_j$. From this bound we can then obtain an upper
bound for the maximal repetition by Theorem
\ref{theoTopEntropyRepetition}.
\begin{proposition}
  \label{theoTopEntropyRHA}
  For the RHA process, almost surely
  \begin{align}
    H_{top}(2^m|\mathcal{X})\le 2\log k_{m}
    .
  \end{align}
\end{proposition}
\begin{proof}
  For a given realization of the RHA process (i.e., for fixed
  $Y_j^m$), there are at most $k_m$ different values of blocks
  $X^m_j$. Therefore, there are at most $k_m^2$ different values of
  blocks $X^m_{kj}$ in sequence $\mathcal{X}$.
\end{proof}

Obtaining a lower bound for the topological entropy and an upper bound
for the maximal repetition of blocks $X^n_j$ is more involved. These
topics will be discussed in the following sections. For this goal, we
will consider events $A_{n,-1}:=\emptyset$ and
\begin{align}
  \label{NoRepeat}
  A_{nm}:=(\text{$X^n_1$ consists of $2^{n-m}$ distinct blocks $X^m_j$})
\end{align}
We have $P(A_{nn})=1$ and $A_{nm}\supset A_{n,m-1}$. Probabilities
$P(A_{nm})$ will be called probabilities of no repeat.
\begin{proposition} 
  \label{theoUpperRepetition}
  For the RHA process, we have $P(A_{nm})=0$ for $k_m<2^{n-m}$,
  whereas for $k_m\ge 2^{n-m}$ and $m<n$ we have
  \begin{align}
    \label{NoRepeatEx}
    P(A_{nm})=P(A_{n,m+1})
    \frac{k_m(k_m-1)\ldots(k_m-2^{n-m}+1)}{k_m^2(k_m^2-1)\ldots(k_m^2-2^{n-m-1}+1)}
    .
  \end{align}
\end{proposition}
\begin{proof}
  There are no more than $k_m$ distinct blocks $X^m_j$ in block
  $X^n_1$. Thus $P(A_{nm})=0$ for $k_m<2^{n-m}$. Now assume $k_m\ge
  2^{n-m}$.  Introduce random variables $D_{mi}$ such that
  $X^n_1=Y^m_{D_{m1}}\times ...\times Y^m_{D_{m2^{n-m}}}$.  Consider
  probabilities $p_m=P(D_{m1}=d_1,...,D_{m2^{n-m}}=d_{2^{n-m}})$,
  where $d_i$ are distinct. It can be easily shown by induction on
  decreasing $m$ that $p_m$ do not depend on $d_i$ and satisfy
  \begin{align*}
    p_m=p_{m+1}
    \binom{k_m^2}{k_{m+1}}^{-1}\binom{k_m^2-2^{n-m-1}}{k_{m+1}-2^{n-m-1}}
    .
  \end{align*}
  Moreover, since $p_m$ do not depend on $d_i$, we obtain
  $P(A_{nm})=p_m k_m(k_m-1)\ldots(k_m-2^{n-m}+1)$. Hence the claim
  follows.
\end{proof}

\section{Shannon block entropy}
\label{secEntropy}

This section is the last preparatory section. Here we will bound the
Shannon entropies of blocks $X^n_j$ in terms of perplexities $k_n$.
To establish some necessary notation, for random variables $X$, $Y$
and $Z$, where $X$ is discrete whereas $Y$ and $Z$ need not be so,
besides Shannon entropy $H(X)=\sred_P\kwad{-\log P(X)}$, we define
conditional entropy $H(X|Y)=\sred_P\kwad{-\log P(X|Y)}$, mutual
information $I(X;Y):=H(X)-H(X|Y)$, and conditional mutual information
$I(X;Y|Z):=H(X|Z)-H(X|Y,Z)$.  Given these objects, we will bound the
Shannon entropies of blocks of the RHA process.

The first result is a corollary of Proposition \ref{theoKpairs}, which
says that conditional entropy of blocks $X^n_j$ given the entire pool
of admissible blocks of the same length $\mathcal{G}_{\le n}$ is exactly
equal to the logarithm of perplexity.
\begin{proposition}
  \label{theoCondEntropy}
  We have 
  \begin{align}
    \label{CondEntropy}
    H(X^n_j|\mathcal{G}_{\le n})= \log k_n
  \end{align}
  and $I(X^n_j;X^n_{j+1}|\mathcal{G}_{\le n})=0$.
\end{proposition}
\begin{proof}
  Given $\mathcal{G}_{\le n}$, the correspondence between $X^n_j$ and
  $K_{nj}$ is one-to-one. Hence $H(X^n_j|\mathcal{G}_{\le
    n})=H(K_{nj}|\mathcal{G}_{\le n})$.  From Proposition \ref{theoKpairs}
  we further obtain $H(K_{nj}|\mathcal{G}_{\le n})=H(K_{nj})=\log k_n$
  and $H(K_{nj},K_{n,j+1}|\mathcal{G}_{\le
    n})=H(K_{nj})+H(K_{n,j+1})$.
\end{proof}

The second result is an exact expression for the Shannon entropy of the pool
of admissible blocks $\mathcal{G}_{\le n}$, also in term of
perplexities.
\begin{proposition}
  \label{theoEntropyI}
  We have 
  \begin{align}
    \label{EntropyI}
    H(\mathcal{G}_{\le n})=\sum_{l=1}^n \log
    \binom{k_{l-1}^2}{k_l}
    .
  \end{align}
\end{proposition}
\begin{proof}
  The claim follows by chain rule $H(\mathcal{G}_{\le
    n})=H(\mathcal{G}_{\le n-1})+H(\mathcal{G}_{\le
    n}|\mathcal{G}_{\le n-1})$ from $H(\mathcal{G}_{\le 0})=0$ and
  $H(\mathcal{G}_{\le n}|\mathcal{G}_{\le n-1})=\log
  \binom{k_{n-1}^2}{k_n}$.
\end{proof}

Combining the above two results, we can provide an upper bound for the
unconditional Shannon entropy of blocks $X^n_j$.
\begin{proposition}
  \label{theoTwoPartEntropy}
  We have 
  \begin{align}
    \label{TwoPartEntropy}
    H(X^n_j)\le \min_{0\le l\le n} \okra{H(\mathcal{G}_{\le l}) +
      2^{n-l}\log k_l}
    .
  \end{align}
\end{proposition}
\begin{proof}
  For any $0\le l\le n$ we have $H(X^n_j)\le H(X^n_j,\mathcal{G}_{\le
    l}) = H(X^n_j|\mathcal{G}_{\le l})+H(\mathcal{G}_{\le l})$,
  whereas $H(X^n_j|\mathcal{G}_{\le l})\le
  2^{n-l}H(K_{lj}|\mathcal{G}_{\le l})=2^{n-l}H(K_{lj})=2^{n-l}\log
  k_l$.
\end{proof}

Given Propositions \ref{theoCondEntropy} and \ref{theoTwoPartEntropy}, we
may introduce an important parameter of the RHA process, which we will
call the combinatorial entropy rate.
\begin{definition}
  The combinatorial entropy rate of the RHA process is
  \begin{align}
    h:=\inf_{l\in\mathbb{N}} 2^{-l}\log k_l=
    \lim_{l\rightarrow\infty} 2^{-l}\log k_l
    .
  \end{align}
\end{definition}
\begin{proposition}
  \label{theoEntropyRate}
  We have 
  \begin{align}
    \inf_{n\in\mathbb{N}} 2^{-n}H(X^n_j)=h
    .
  \end{align}
\end{proposition}
\begin{proof}
  On the one hand, by Proposition \ref{theoCondEntropy},
  \begin{align*}
    \inf_{n\in\mathbb{N}} 2^{-n}H(X^n_j)\ge \inf_{n\in\mathbb{N}}
    2^{-n} H(X^n_j|\mathcal{G}_{\le n})=\inf_{l\in\mathbb{N}} 2^{-l}\log
    k_l
    .
  \end{align*}
  On the other hand, by Proposition \ref{theoTwoPartEntropy},
  \begin{align*}
    \inf_{n\in\mathbb{N}} 2^{-n}H(X^n_j)\le \inf_{l\in\mathbb{N}}
    \inf_{n\in\mathbb{N}} \okra{2^{-n} H(\mathcal{G}_{\le l}) +
      2^{-l}\log k_l}=\inf_{l\in\mathbb{N}} 2^{-l}\log k_l
    .
  \end{align*}
\end{proof}

Proposition \ref{theoEntropyRate} combined with Proposition
\ref{theoBlockEntropyAMS} yields a bound for the Shannon entropy rate
of the stationary mean of the RHA process.
\begin{proposition}
  \label{theoEntropyRateAMS}
  For the stationary mean $\mu$ of the RHA process, we have
  \begin{align}
    h/2\le h_\mu\le 2h
    .
  \end{align}
\end{proposition}
\begin{proof}
  Divide inequality (\ref{BlockEntropyAMS}) by $2^n$ and take the
  infimum.
\end{proof}
In particular, the combinatorial entropy rate vanishes ($h=0$) if and
only if the Shannon entropy rate of the stationary mean vanishes
($h_\mu=0$) as well. This happens in particular for perplexities
(\ref{HilbergPerplexity}).

Inequality $H(X^n_j)\ge H(X^n_j|\mathcal{G}_{\le n})=\log k_n$ gives a
certain lower bound for the Shannon block entropy of the RHA process. For
perplexities (\ref{HilbergPerplexity}), this lower bound is orders of
magnitude smaller than the upper bound
(\ref{TwoPartEntropy}). Concluding this section we would like to
produce a lower bound which is of comparable order to
(\ref{TwoPartEntropy}).
\begin{proposition}
We have
  \begin{align}
    \label{TwoPartEntropyII}
    H(X^n_j)\ge \max_{0\le l\le n} \okra{\log\binom{k_{l-1}^2}{k_l}
      -\log\binom{k_{l-1}^2-2^{n-l}}{k_l-2^{n-l}}}P(A_{nl})
    ,
  \end{align}
where $P(A_{nl})$ are the probabilities of no repeat (\ref{NoRepeatEx}).
\end{proposition}
\begin{proof}
  We have 
  \begin{align*}
    H(X^n_j)\ge I(X^n_j;\mathcal{G}_{\le l}|\mathcal{G}_{\le l-1})=
    H(\mathcal{G}_{\le l}|\mathcal{G}_{\le l-1})-
    H(\mathcal{G}_{\le l}|\mathcal{G}_{\le l-1},X^n_j) 
    .
  \end{align*}
  We have $H(\mathcal{G}_{\le l}|\mathcal{G}_{\le l-1})=\log
  \binom{k_{l-1}^2}{k_l}$. As for $H(\mathcal{G}_{\le
    l}|\mathcal{G}_{\le l-1},X^n_j)$, we may propose the following
  bound. Given $X^n_j$ consisting of $2^{n-l}$ distinct blocks of
  length $2^l$, tuple $(L_{lj},R_{lj})_{j\in\klam{1,...,k_l}}$ may
  assume at most $\binom{k_{l-1}^2-2^{n-l}}{k_l-2^{n-l}}$ distinct
  values. Hence
  \begin{align*}
    H(\mathcal{G}_{\le l}|\mathcal{G}_{\le l-1},X^n_j)
    \le P(A_{nl})\log\binom{k_{l-1}^2-2^{n-l}}{k_l-2^{n-l}}
    ,
  \end{align*}
  from which the claim follows.
\end{proof}

\section{Main result}
\label{secMain}

Now we can demonstrate the main result, which will conclude our paper.
\begin{proof*}{Theorem \ref{theoMain}}
  \begin{enumerate}
  \item For perplexities (\ref{HilbergPerplexity}) the combinatorial
    entropy rate is $h=0$. Hence $h_\mu=0$ by Proposition
    \ref{theoEntropyRateAMS}.
  \item By (\ref{EntropyI}), entropy $H(\mathcal{G}_{\le n})$ can be
    bounded as
    \begin{align*}
      H(\mathcal{G}_{\le n})&=\sum_{l=1}^n \log
      \binom{k_{l-1}^2}{k_l}
      \le \sum_{l=1}^n 2k_l\log k_{l-1}
      \le 2nk_n\log k_n
      .
    \end{align*}
    Hence, from (\ref{TwoPartEntropy}), for $0\le l\le n$ 
    we obtain an upper bound:
    \begin{align*}
      H(X^n_j)&\le \okra{2lk_l + 2^{n-l}}\log k_l 
      .
    \end{align*}
    If we choose $l=\floor{\beta^{-1}\log_2\okra{\frac{n\log 2}{\log
          n}}}$ then for perplexities (\ref{HilbergPerplexity}) we obtain
    \begin{align}
      H(X^n_j)&\le \kwad{2\beta^{-1}\log_2\okra{\frac{n\log 2}{\log n}} 2^{n/\log
          n} + 2^n\okra{\frac{n\log 2}{\log n}}^{-1/\beta}}
      \frac{n\log 2}{\log n}
      \nonumber\\
      \label{HUpper}
      &=\Theta\okra{2^{n}\okra{\frac{\log n}{n}}^{1/\beta-1}}
      .
    \end{align}
    
    On the other hand, from (\ref{TwoPartEntropyII}) and (\ref{NoRepeatEx}),
    for $0\le l\le n$ we have 
    \begin{align*}
      H(X^n_j)&\ge 
      \okra{\log\binom{k_{l-1}^2}{k_l}
        -\log\binom{k_{l-1}^2-2^{n-l}}{k_l-2^{n-l}}}P(A_{nl})
      \\
      &\ge
      2^{n-l}\log\okra{\frac{k_{l-1}^2-2^{n-l}+1}{k_l-2^{n-l}+1}}P(A_{nl})
      ,
    \end{align*}
    where
    \begin{align}
      P(A_{nl})&=
      \prod_{m=l}^{n-1}
      \frac{k_m(k_m-1)\ldots(k_m-2^{n-m}+1)}{k_m^2(k_m^2-1)\ldots(k_m^2-2^{n-m-1}+1)}
      \nonumber
      \\
      &\ge
      \prod_{m=l}^{n-1}
      \okra{\frac{(k_m-2^{n-m}+2)(k_m-2^{n-m}+1)}{k_m^2-2^{n-m-1}+1}}^{2^{n-m-1}}
      \nonumber
      \\
      &\ge
      \okra{\frac{(k_l-2^{n-l}+2)(k_l-2^{n-l}+1)}{k_l^2-2^{n-l-1}+1}}^{\sum_{m=l}^{n-1}
        2^{n-m-1}}
      \nonumber
      \\
      &\ge
      \okra{1-\frac{k_l(2^{n-l+1}-3)+2}{k_l^2-2^{n-l-1}+1}}^{2^{n}}
      \nonumber
      \\
      \label{UpperNoRepeat}
      &\ge
      1-2^{n}\frac{k_l(2^{n-l+1}-3)+2}{k_l^2-2^{n-l-1}+1}
      .
    \end{align}
    If we choose $l=\ceil{\beta^{-1}\log_2(2n)}$ then for perplexities
    (\ref{HilbergPerplexity}) we obtain that
    $k_l>\exp(2n)>2^{2n}$. Hence $P(A_{nl})$ is greater than a certain
    constant $\alpha>0$ and
    \begin{align}
      \label{HLower}
      H(X^n_j)&\ge \alpha 2^n(2n)^{-1/\beta}[2^{1-\beta}-1]2n
      =\Theta\okra{2^{n}\okra{\frac{1}{n}}^{1/\beta-1}}
      .
  \end{align}
  By (\ref{HUpper}) and (\ref{HLower}), from Proposition
  \ref{theoBlockEntropyAMS}, we obtain the desired sandwich bound for
  the entropy of the stationary mean.

\item By Proposition \ref{theoTopEntropyRHA} and Proposition \ref{theoPhiAMS} we
  obtain
  \begin{align*}
    0&=\sred_P\boole{H_{top}(2^m|\mathcal{X})>2\log k_m}
    \\
    &\ge \sred_P\boole{H_{top}(2^m|X^{n+1}_j)>2\log k_m}
    \\
    &\ge \sred_\mu\boole{H_{top}(2^m|\xi_{1:2^n})>2\log k_m}
    .
  \end{align*}
  Hence $\mu$-almost surely $H_{top}(2^m|\xi_{1:\infty})\le 2\log
  k_m=2^{\beta m+1}$, which implies the upper bound
  $H_{top}(m|\xi_{1:\infty})<C_1 m^\beta$ for a certain constant
  $C_1$. From this we obtain the lower bound $L(\xi_{1:m})>C_2(\log
  m)^{1/\beta}$ by Theorem \ref{theoTopEntropyRepetition}.

  As for the converse bounds, we have $L(X^n_1)\ge 2^{l}$ for
  $A_{nl}^c$, where $A_{nl}$ are the events of no repeat
  (\ref{NoRepeat}). Hence by Proposition \ref{theoPhiAMS},
  \begin{align*}
    \sred_\mu\boole{L(\xi_{1:2^n})\ge l}
    \le \sred_P\boole{L(X^{n+1}_l)\ge l}
    \le 1-P(A_{n+1,l})
    .
  \end{align*}
  Now, if we choose $l=\ceil{\beta^{-1}\log_2(2n)}$ then for
  perplexities (\ref{HilbergPerplexity}) we obtain that
  $k_l>\exp(2n)>2^{2n}$. Hence, by (\ref{UpperNoRepeat}),
  $\sum_{n=0}^{\infty}(1-P(A_{n+1,l}))<\infty$. Consequently, by the
  Borel-Cantelli lemma $L(\xi_{1:2^n})< l$ must hold for sufficiently
  large $n$ $\mu$-almost surely.  Thus $L(\xi_{1:m})<C_3(\log
  m)^{1/\beta}$ for sufficiently large $m$.  From this we obtain the
  lower bound $H_{top}(m|\xi_{1:\infty})>C_4 m^\beta$ for sufficiently
  large $m$ by Theorem \ref{theoTopEntropyRepetition}.

\item Denote the random ergodic measure $F=\mu(\cdot|\mathcal{I})$ of
  the stationary mean $\mu$.  The entropy of the shift-invariant
  algebra with respect to $\mu$ may be bounded by mutual information as
  \begin{align*}
    H_\mu(\mathcal{I})=
    \lim_{m\rightarrow\infty} I_\mu(\mathcal{I};\xi_{1:m})
    &=
    \lim_{m\rightarrow\infty}
    \kwad{H_\mu(\xi_{1:m})-H_\mu(\xi_{1:m}|\mathcal{I})}
    \\
    &=
    \lim_{m\rightarrow\infty} \kwad{H_{\mu}(m)-\sred_\mu H_F(m)}
    \\
    &=
    \lim_{m\rightarrow\infty} \kwad{H_{\mu}(m)-\sred_\mu H_{top}(m|\xi_{1:\infty})}
    =
    \infty
    .
  \end{align*}
  Since the entropy of the shift-invariant algebra is strictly
  positive, the measure $\mu$ is nonergodic.
\end{enumerate}
\end{proof*}

\section*{Acknowledgment}

The author wishes to thank Jan Mielniczuk and an anonymous referee for
valuable comments and literature suggestions.

\bibliographystyle{IEEEtran}

\bibliography{0-journals-abbrv,0-publishers-abbrv,ai,mine,tcs,ql,books}

\end{document}